\documentclass[12pt]{article}
\usepackage[english]{babel}
\usepackage[utf8]{inputenc}
\usepackage[T1]{fontenc}
\usepackage{listings}
\usepackage{amsbsy}
\usepackage{amsfonts}
\usepackage{amsmath}
\usepackage{amssymb}
\usepackage{amsthm}
\usepackage{graphicx}
\usepackage[pdftex]{color}
\usepackage{dsfont}
\usepackage{enumerate}
\usepackage{algorithm,algcompatible}
\usepackage{threeparttable, booktabs, graphicx, float, rotating, tikz}
\usepackage{float}
\usepackage{calc, color, setspace}%
\usepackage{indentfirst}
\usepackage{longtable}
\usepackage{multirow}
\usepackage{natbib}
\usepackage{xr-hyper}
\usepackage{hyperref,bookmark}

\usepackage{fullpage}

\newtheorem{theorem}{Theorem}[section]
\newtheorem{lemma}[theorem]{Lemma}

\newtheorem{remark}[theorem]{Remark}
\newtheorem{proposition}{Proposition}
\newtheorem{definition}{Definition}

\hypersetup{
  colorlinks=true,
  linkcolor=blue,
  citecolor=red,
  filecolor=blue,
  urlcolor=blue,
}
\newcommand{\blind}{1}

\protect

\begin{document}

\def\spacingset#1{\renewcommand{\baselinestretch}%
{#1}\small\normalsize} \spacingset{1}


\if1\blind
{
\title{\bf Generalized inverse-Gaussian frailty models with application to TARGET neuroblastoma data}
\author{Luiza S.C. Piancastelli$^{\star}$$^\sharp$\footnote{E-mail: \texttt{luiza.piancastelli@ucdconnect.ie}},\,\, Wagner Barreto-Souza$^\ast$$^\sharp$\footnote{Corresponding author. E-mail: \texttt{wagner.barretosouza@kaust.edu.sa}}\,\, and\, Vin\'icius D. Mayrink$^\sharp$\footnote{E-mail: \texttt{vdm@est.ufmg.br}}\\
	\small$^\star$\it School of Mathematics and Statistics, University College Dublin, Dublin, Ireland\\
	\small$^\ast$\it Statistics Program, King Abdullah University of Science and Technology, Thuwal, Saudi Arabia\\
	\small$^{\sharp}$\it Departamento de Estat\'istica, Universidade Federal de Minas Gerais, Belo Horizonte, Brazil}
  \maketitle
} \fi
\if0\blind
{
  \bigskip
  \bigskip
  \bigskip
  \begin{center}
    {\LARGE\bf Generalized inverse-Gaussian frailty models with application to TARGET neuroblastoma data}
\end{center}
  \medskip
} \fi

\bigskip
\addtocontents{toc}{\protect\setcounter{tocdepth}{1}}

\begin{abstract}
A new class of survival frailty models based on the Generalized Inverse-Gaussian (GIG) distributions is proposed. We show that the GIG frailty models are flexible and mathematically convenient like the popular gamma frailty model. Furthermore, our proposed class is robust and does not present some computational issues experienced by the gamma model. By assuming a piecewise-exponential baseline hazard function, which gives a semiparametric flavour for our frailty class, we propose an EM-algorithm for estimating the model parameters and provide an explicit expression for the information matrix. Simulated results are addressed to check the finite sample behavior of the EM-estimators and also to study the performance of the GIG models under misspecification. We apply our methodology to a TARGET (Therapeutically Applicable Research to Generate Effective Treatments) data about survival time of patients with neuroblastoma cancer and show some advantages of the GIG frailties over existing models in the literature.
\end{abstract}

\vfill

\section{Introduction}
\label{sec1}
When dealing with time to event data, the most popular statistical approach is the proportional hazards model by \cite{cox72}. This model is based on the hazard function and accommodates well censored and truncated data, which are key elements in survival analysis. 

One situation in which the proportional hazards model can be deficient occurs when unobserved sources of heterogeneity are present in the data. This might be explained by the lack of important covariates in the study, which are difficult to measure or were not collected because the researcher did not know its importance in the first place. In this case, the deficiency of the proportional hazards model is the assumption of a homogeneous population.  Another common situation in which the proportional hazards model is problematic occurs when there is correlated survival data. The correlation arises, for example, when repeated measures are collected for each individual or when some common traits such as biological or environmental factors are shared.

The described situations are usually treated by assuming a frailty model configured as a natural extension to the Cox model. This approach introduces a latent random component that acts multiplicatively in the hazard function for an individual or a group of individuals. The univariate frailty modeling can handle unobserved sources of heterogeneity and independence is assumed among individual frailties. One of the multivariate versions of this methodology known as shared frailty model was introduced by \cite{clayton}, which was motivated by the analysis of familial tendency in disease incidence. The author assumed that individuals in the same group share the frailty term and a positive dependence between those individuals is created.

In both univariate and multivariate versions, fitting the unobserved risk components is of most importance to properly evaluate the covariate effects. Some distributions in $\mathbb{R^+}$ are commonly assumed for the frailty term having certain desired mathematical properties, since the inferential aspects of the frailty models pose additional difficulties in comparison with usual mixed models due to censoring and truncation. The gamma distribution is the most common choice for this task. This is a model that became very popular due to its mathematical convenience, being explored by several authors such as \cite{vaupel}, \cite{oakes1}, \cite{oakes2}, \cite{klein} and \cite{yashin}. Other well known options in this field are the parametric inverse-Gaussian \citep{houg1984}, positive stable \citep{houg1986}, log-normal \citep{mcg} and power variance family \citep{cro1989,houetal1992} frailty models.

Semiparametric versions of a frailty model are often preferred in the literature, since it allows the estimation of regression effects without the need to explicitly impose a particular form for the baseline hazard function. Assuming a parametric formulation for this function can be restrictive, since the analysis is bounded by the possible shapes of the chosen function. In addition, the parametric formulation can be difficult to identify or to test for adequacy. There are different methods to develop semiparametric versions of frailty models. For instance, the approach introduced by \cite{klein}, which is based on a modified EM algorithm involving the Cox proportional hazards model. Other possible strategies are the penalized partial likelihood functions introduced by \cite{therneau}, piecewise constant hazards \citep{kim} with raising number of pieces and splines, among others. In this paper, we will focus on the piecewise constant hazards approach. Although this is a parametric model at all, it is quite flexible since we do not have to assume any specific form for the baseline hazard function, which gives a semiparametric flavour. For a discussion about the advantages of this choice, we refer \cite{lawless}. Other contributions on frailty modeling are due to \cite{balpen2006}, \cite{klein2012}, \cite{callegaro}, \cite{farrington}, \cite{unkfar2012}, \cite{chenetal2013}, \cite{enketal2014}, \cite{moletal2015}, \cite{puthou2015}, \cite{balpal2016}, \cite{leaetal2017}, \cite{unk2017} and \cite{barretosouza}, just to name a few. For a good account on frailty models, we recommend the books by \cite{hou2000}, \cite{wienke} and \cite{han2019}.

Our chief goal in this paper is to introduce a new class of frailty models based on the generalized inverse-Gaussian (GIG) distributions, which has
some advantages over existing models in the literature, as it will be shown along the paper. Our proposal praises for efficiency in terms of computation and flexibility  without compromising mathematical tractability, since all the main expressions have closed forms, like in the gamma model.

It is worth to mention that both GIG and power variance function (PVF) families are infinitely divisible (for instance, see respectively \cite{barhal1977} and \cite{han2019}), so that a natural question can arise: Is there some connection between them? In fact, they share the inverse-Gaussian model as a common member (in our case this model is obtained by taking $\lambda=-1/2$; see Remark \ref{gridlambda}) but the other cases are different; the Laplace transform of the GIG distributions given in Eq. (\ref{laplace}) of this paper is different from the PVF Laplace transform given in Subsection 5.5 from \cite{han2019}.

We now highlight some important contributions of this present paper as follows. {\it (i) Robustness.} In simulation studies we show that the GIG distribution ensures the flexibility we seek, as the model can return optimal results when data is generated with different frailty distributions. {\it (ii) Mathematical tractability.}  We provide closed forms for the unconditional density, survival and hazard functions related to the generalized inverse-Gaussian frailty model. Further, the conditional distribution of the frailty given the data is again GIG distributed. This kind of conjugacy, also true for the gamma case, is attractive and it allows us to provide an EM-algorithm with closed E-step. Therefore, our proposed frailty models enjoy the same mathematical and analytical tractability of the gamma model. We are unaware of other existing frailty model having all these features. {\it (iii) Cluster survival data.} Although the neuroblastoma data considered here is not clustered, we introduce our model in a more general setting allowing clustered data analysis. {\it (iv) Computational tractability.} Some numerical problems are experienced in the data analysis through gamma frailty model, in contrast with our proposed class of GIG models.

This paper is organized as follows. In Section \ref{spec}, we introduce the class of GIG frailty models and obtain some basic results. In Section \ref{clust}, we discuss maximum likelihood estimation under the parametric case and propose an EM-algorithm \citep{dempster} for estimating the parameters with focus in the piecewise exponential baseline assumption, where the key expressions and a detailed description are provided. In Section \ref{sec:simulation}, we present a comprehensive Monte Carlo simulation study, for the EM-approach, where the robustness of the GIG frailty model under misspecification is compared with other existing models in literature. Statistical analysis based on GIG frailty models to a TARGET (Therapeutically Applicable Research to Generate Effective Treatments) data about survival time of patients with neuroblastoma cancer is addressed in Section \ref{sec:analysis}. Discussion about our findings and some points to be attacked in future works are presented in Section \ref{discussion}.


\section{Model specification and basic results} \label{spec}

This section introduces the GIG frailty models. We begin by discussing the case without covariates. The inclusion of covariates will be addressed in next section. The Generalized Inverse-Gaussian distribution with parameters $\lambda \in \mathbb{R}$, $a > 0$ and $b>0$ has the following density
\begin{eqnarray} \label{gig_density}
g(x) = \frac{(a/b)^{\lambda/2}}{2 \ K_\lambda(\sqrt{ab})}x^{\lambda-1} \exp\{ -(ax + bx^{-1})/2\}, \quad \mbox{for} \quad x > 0,
\end{eqnarray}
where $K_\lambda(x) = \int_{0}^{\infty} u^{\lambda-1} \exp\{-\frac{t}{2}(u+u^{-1})\} \ du$ denotes the third kind modified Bessel function. A random variable $X$ with density function (\ref{gig_density}) is denoted by $X\sim\mbox{GIG}(a,b,\lambda)$.

Some basic properties of the GIG distribution are presented next. The Laplace transform associated to (\ref{gig_density}) is given by
\begin{eqnarray} \label{laplace}
L(t) = \frac{K_{\lambda}(\sqrt{(a+2t)b})}{K_{\lambda}(\sqrt{ab})} \left({\frac{a}{a+2t}} \right)^{\lambda/2}, \quad \mbox{for} \quad t>-a/2. 
\end{eqnarray}

The moments of a GIG distributed random variable $X$ are given by
\begin{eqnarray}\label{moments}
E(X^k) = \frac{K_{\lambda+k}(\sqrt{ab})}{K_{\lambda}(\sqrt{ab})}(b/a)^{k/2}, \quad k\in\mathbb R.
\end{eqnarray}

Having a closed form for the Laplace transform is an advantage in the construction of frailty models. This is justified by the fact that the marginal survival and density functions of the frailty model can be found through this expression. The fact that the Laplace transform of the gamma distribution has a closed form is one of the reasons for the popularity of such model. The main structure of our proposed class of frailty models is given in what follows.

\begin{definition} \label{def_frailty}
	Assume $Z$ is a random variable following a GIG distribution with parameters $a=\alpha^{-1}$, $b=\alpha^{-1}$ and $\lambda>0$, where $\alpha>0$; we denote $Z\sim \mbox{GIG}(\alpha^{-1},\alpha^{-1},\lambda)$. Consider the time to event of an individual to be denoted by a random variable $T$, where the latent effect $Z$ acts multiplicatively in the baseline hazard function. Conditional on $Z$, the hazard function is given by $h(t|Z) = Z \ h_0(t)$. Here $h_0(t)$ is the baseline hazard function and $Z$ is the frailty variable. 
\end{definition}
\begin{remark}\label{gridlambda}
	A motivation for assuming the GIG distribution for the frailty term is that, for a fixed $\lambda$, there are several known special cases of this distribution. The inverse-Gaussian (IG), reciprocal inverse-Gaussian (RIG), hyperbolic (HYP) and positive hyperbolic (PHYP) frailty models are obtained by choosing $\lambda=-1/2$, $\lambda=1/2$, $\lambda=0$ and $\lambda=1$, respectively. Hence, by changing the value of $\lambda$, one can fit different frailty distributions to the data. 
\end{remark}

\begin{remark}\label{specialcases}
	The parameterization considered for our frailty class is such that the IG case ($\lambda=-1/2$) satisfies $E(Z) = 1$ and $\mbox{Var}(Z) = \alpha$. This is the parameterization considered for IG frailty model in the literature; see \cite{wienke}.
\end{remark}

Through known results on frailty modeling, we can obtain the marginal survival function $S(\cdot)$ and density $f(\cdot)$ of $T$. In line with the Definition \ref{def_frailty}, we denote by $h_0(t)$ the baseline hazard function and $H_0(t) = \int_{0}^{t} h_0(u) \ du$ the cumulative baseline hazard function. The marginal survival function is 
\begin{eqnarray}\label{surv}
S(t) = L(H_{0}(t)) = \left( \frac{1}{1 + 2\alpha H_{0}(t)} \right)^{\lambda/2} \frac{K_{\lambda}(\sqrt{\alpha^{-1}\left[\alpha^{-1} + 2H_0(t)\right]})}{K_{\lambda}(\alpha^{-1})}, \quad \mbox{for} \quad t>0.
\end{eqnarray}

From (\ref{surv}) and by using the recurrence identity of the Bessel function $K_v(z) = \dfrac{z}{2v}[K_{v+1}(z) - K_{v-1}(z)]$, we obtain that the marginal density function assumes the form
\begin{eqnarray*}
	f(t) = \frac{h_0(t)}{{(1+2\alpha H_0(t))}^\frac{\lambda+1}{2}} \frac{K_{\lambda+1}(\sqrt{\alpha^{-1}[\alpha^{-1} + 2H_0(t)]})}{K_{\lambda}(\alpha^{-1})}, \quad \mbox{for} \quad t>0.
\end{eqnarray*}

Next, we discuss another way of characterizing the frailty distribution that was introduced by \cite{farrington}, which is the relative frailty variance (hereafter RFV). The RFV is a measure of how the heterogeneity of the population evolves over time. This function can be used as a way to compare patterns of dependence among different frailty models and is obtained through the Laplace transform of the frailty distribution.
Let $J(s) = \log L(-s)$, where $L(\cdot)$ is the Laplace transform and $\mu$ is the expected value of the frailty. Then, RFV$(s) = J''(-s/\mu)/J'(-s/\mu)^2$; see Eq. (3) from \cite{farrington}. The expressions required to calculate the RFV for the GIG frailty model are given in Appendix A.

\begin{figure}[!h]
	\centering
	\includegraphics[width=0.5\textwidth]{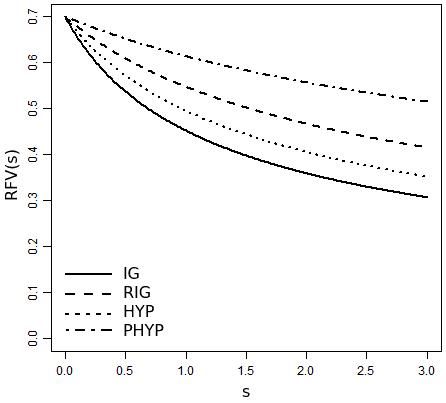}
	\caption{Relative frailty variance evolution for IG, RIG, HYP and PHYP frailties with $\alpha$ chosen such that RFV$(0)=0.7$ in each case.}
	\label{fig_rfv}
\end{figure}

In Figure \ref{fig_rfv}, we present the evolution of the RFV$(s)$ function for the IG, RIG, HYP and PHYP frailties. For the four distributions under comparison, we find for each of them the root of the equation $RFV(0)-0.7 = 0$. That is, the value of $\alpha$ for which RFV(0) is equal to $0.7$. In this way we can better compare how their evolutions differ, since they have the same starting point.

The results of Figure \ref{fig_rfv} evidence similarity within the GIG family regarding the RVF evolution. All special cases have shown decreasing RFV over time, which means that individuals who survive tend to be less heterogeneous (or less dependent) over time. What differentiates among the distributions under comparison is the speed in which this happens. This can be observed in the curvature of the trajectory $RFV(s)$ versus $s$. The most strong curvature is related to the inverse Gaussian model. As a consequence, in this frailty distribution the decrease in heterogeneity (or dependence) occurs more quickly. The PHYP distribution shows a slower decay. The HYP and RIG models have intermediate behaviors for the evolution of the relative frailty variance.

We end this section by calling attention to the fact that the proposed methodology is not restricted to the four options of $\lambda$ discussed in Remark \ref{gridlambda}. These special cases are considered in the simulation studies developed ahead and we believe that they provide enough material to understand the main aspects related to the GIG class. Another strategy indicated for practical applications, as we will consider for the neuroblastoma data analysis, is testing a grid of $\lambda$ values and choosing the one that maximizes the log-likelihood function. In other words, the strategy is to perform a profile likelihood approach for estimating $\lambda$.

\section{GIG frailty model for clustered survival data} \label{clust}

This section explains the GIG frailty model in two contexts: one including a regression structure and the framework extending the analysis to the multivariate case. In the univariate approach, each individual has its own frailty value. This version is applicable, for example, when important covariates are missing from the analysis or when individual heterogeneity is naturally present. On the other hand, in shared frailty models, individuals in a group share the frailty value, which creates a positive dependence among them. With that, we are able to model correlated or clustered survival data. Evidently, this can be reduced to the univariate approach when groups are formed by one observation each. As mentioned previously, although the data considered here is univariate, we introduce our class of frailty models in a more general setting and, therefore, they can be applied for broad situations.

Assume $m$ clusters with the $i$-th cluster having $n_i$ individuals, for $i=1,\cdots,m$. Here $T_{ij}^0$ and $C_{ij}$ denote the failure and censoring times for the individual $j = 1,\cdots,n_i$ in the $i$-th cluster. The total sample size is $n = \sum_{i=1}^{m} n_i$. In addition, let $T_{ij} = \min\{T_{ij}^0,C_{ij}\}$ be the observable response variables and $\delta_{ij}= I\{T_{ij}^0 \le C_{ij}\}$ the failure indicator. Naturally, the frailty $Z_i$ is associated to the $i$-th cluster. In order to complete the model specification, we make the assumptions that given $Z_i$, \{$(T_{ij}^0,C_{ij}); j=1,\cdots,n_i \}$ are conditionally independent and that $T_{ij}^0$ and $C_{ij}$ are independent for all $j$. Another assumption we rely on is that the censoring times within a cluster $\{C_{ij}; j=1,\cdots,n_i \}$ are non-informative with respect to $Z_i$ for all $i$.

Conditional on the frailty, the model has a structure similar to the proportional hazards model in \cite{cox72}. The frailty term $Z_i$ is inserted as follows $h(t_{ij}|Z_i) = Z_ih_0(t_{ij})\exp({x_{ij}^\top \beta})$, for $t_{ij}>0$,
where $x_{ij}$ is the vector of covariates of the $j$-th individual from the $i$-th cluster.

In the forthcoming discussion, we present maximum likelihood estimation for the parametric case, where a specific form for the baseline hazard function needs to be assumed. Further, we propose an EM-algorithm for estimating the model parameters with focus on the piecewise exponential baseline hazard function, which provides a model flexibility since we do not impose any specific shape for that function.

\subsection{Maximum likelihood estimation}

In this section, we present the unconditional survival and density functions to construct the likelihood function and perform maximum likelihood estimation under the parametric approach. The joint survival function of a cluster can be found through the Laplace transform of the frailty variable in (\ref{laplace}). The joint survival function associated to the $i$-th cluster is
\begin{eqnarray} \label{joint_surv}
S(t_{i1},\cdots,t_{i \, n_i}) & = & L_{Z_i}\left(\sum_{j=1}^{n_i} H_0(t_{ij})e^{x_{ij}^\top \beta} \right) \\
& = & {\left( \frac{1}{1+2\alpha\sum_{j=1}^{n_i} H_0(t_{ij})e^{x_{ij}^\top \beta} } \right)}^{\lambda/2} 
\;\;\; \frac{K_{\lambda}\left(\sqrt{\alpha^{-1}[ \alpha^{-1}+ 2\sum_{j=1}^{n_i}H_0(t_{ij})e^{x_{ij}^ \top\beta}] }\right) }{K_\lambda (\alpha^{-1})}, \nonumber
\end{eqnarray}
for $t_{ij} > 0 \;\; \forall i,j$. The joint density associated to this survival function is
\begin{eqnarray} \label{joint_dens}
f(t_{i1},\cdots,t_{i \, n_i}) = \frac{\alpha^{-\lambda}}{K_\lambda(\alpha^{-1})} (2/\alpha)^{n_i} \prod_{j=1}^{n_i} h_0(t_{ij}) \ e^{x_{ij}^\top \beta} \chi^{(n_i)}\left(\alpha^{-1}[\alpha^{-1} + 2H_0(t_{ij})e^{x_{ij}^\top \beta}] \right),
\end{eqnarray}
where $\chi^{(k)}(x) = \dfrac{\partial}{\partial x^k} \dfrac{(-1)^{k}}{x^{\phi/2}}K_\phi(\sqrt{x})$, for $k\in\mathbb N$.

Let $\theta = (\beta, H_0, \alpha)^\top$ be the parameter vector and $\ell(\theta)$ denotes the log-likelihood function. Here, $\lambda$ is assumed to be fixed. Using the expressions in (\ref{joint_surv}) and (\ref{joint_dens}) together with Lemma 1 from Appendix A, we can write the log-likelihood function as 
\begin{eqnarray*} \label{loglik_obs}
	\ell(\theta) =\log\left\{ \prod_{i=1}^{m} \frac{\alpha^{-(\sum_{j=1}^{n_i}\delta_{ij})+\lambda}}{K_\lambda(\alpha^{-1})} \; \Psi_{\lambda+\sum_{j=1}^{n_i}\delta_{ij}}\left(\alpha^{-1}[\alpha^{-1} + 2\sum_{j=1}^{n_i}H_{0}(t_{ij})e^{x_{ij}^\top \beta}]\right)\prod_{j=1}^{n_i} \left[h_0(t_{ij})e^{x_{ij}^\top \beta} \right]^{\delta_{ij}}\right\},
\end{eqnarray*}
where $\Psi_\lambda(x) = K_\lambda(\sqrt{x}) / x^{\lambda/2}$.
In order to fit a parametric GIG frailty model, we specify the baseline hazard function $h_0(\cdot)$ and obtain the parameter estimates by maximizing the log-likelihood function through numerical optimization methods. Some of them are available in \verb|R| through the function \verb|optim|; this includes the BFGS and Nelder-Mead methods \citep{fletcher}. In this paper, we assume a Weibull baseline hazard in the parametric approach; this leads to the following formulations $h_0(t) = \sigma \gamma t^{\gamma-1}$ and $H_0(t) = \sigma t^{\gamma}$, for $t>0$ and $\sigma,\gamma>0$.

\subsection{EM-algorithm}\label{semi}
We here propose an EM (Expectation-Maximization) algorithm \citep{dempster} for estimating the parameters. Further, we follow \cite{lawless} for the choice of $h_0(\cdot)$ to be a piecewise-constant hazard function (also known as piecewise exponential hazard function). These authors argue that the use of the piecewise-constant hazard function avoids many problems associated with nonparametric and semiparametric methods for incomplete survival data, but still provides a high degree of robustness. We will denote this approach by PE-GIG model, where PE stands for piecewise exponential baseline hazard function.

The baseline hazard function of a piecewise exponential distribution is given by
\begin{eqnarray} \label{pwe}
h_0(t) = \eta_l \quad \mbox{for} \quad t^{(l-1)} \leq t < t^{(l)} \quad \mbox{and} \quad l = 1,\cdots,k+1 \ ,
\end{eqnarray}
where $t^{(l)}$ denotes the $l$-th ordered time and $0 = t^{(0)} < \min\{t_i; i=1,\cdots,n\} = t^{(1)}  < \cdots < t^{(k)} \leq \max\{t_i; i=1,\cdots,n\} = t^{(k+1)}$ are the pre-specified cut points. This implies that the cumulative baseline function is $H_0(t|\eta)= \sum_{j=1}^{i} \eta_j [\min\{t, t^{(j+1)}\}-t^{(j)}]$ for $t^{(i)} \leq t < t^{(i+1)}$ and $i = 0,\cdots,k$. 

\cite{kim} introduced the estimation of the survival function through the piecewise exponential estimator; they considered as many partitions as the number of observed failures. However, we will follow the steps in \cite{lawless} that showed, through simulation studies, that frailty models based on the piecewise constant hazards often provides satisfactory results for estimating parameters when 8-10 pieces are admitted. 

The estimation of parameters will be done via an EM algorithm. This type of algorithm is appropriate to handle the presence of missing/latent data in the likelihood function and has been used in the literature of frailty models by \cite{klein}, \cite{klein2012}, \cite{callegaro} and more recently by \cite{barretosouza}, just to name a few. In the Expectation (E) step, we compute the conditional expectation of the complete log-likelihood function given the observed data, evaluated at the current parameter estimates; this is called $Q$-function. In the Maximization (M) step, we maximize the $Q$-function. The new estimates obtained in the M step are used to update the $Q$-function. We iterate between the these steps until a pre-specified convergence criterion is reached. 

The complete data set is denoted by $(t_{ij},\delta_{ij},Z_i)$, for $j=1,\cdots,n_i$ and $i=1,\cdots,m$. We observe the pairs $(t_{ij},\delta_{ij})$ and the $Z_i$'s are the latent random effects. The complete log-likelihood function can be written as $\ell_c(\theta) = \ell_1(\beta,H) + \ell_2(\alpha)$, where
\begin{eqnarray*}
	\ell_1(\beta,H)  \propto  \sum_{i=1}^{m} \sum_{j=i}^{n_i} \delta_{ij}\left(x_{ij}^\top \beta + \log h_0(t_{ij})  \right) - \sum_{i=1}^{m} \sum_{j=i}^{n_i} Z_i H_0(t_{ij})e^{x_{ij}^\top \beta}
\end{eqnarray*}
and
\begin{eqnarray*}
	\ell_2(\alpha) \propto -m  \log K_\lambda(\alpha^{-1}) - \frac{1}{2\alpha} \sum_{i=1}^{m} \left(z_i + z_i^{-1} \right).
\end{eqnarray*}

In order to determine the E-step, we need to find the conditional distribution of $Z_i$ given the observed data $\left\{(t_{ij}, \delta_{ij})\right\}_{j=1}^{n_i}$. By using basic properties of conditional densities, it can be shown that this conditional distribution is a GIG law. More specifically, we have that
\begin{eqnarray}\label{conddist}
Z_i|\left\{(t_{ij}, \delta_{ij})\right\}_{j=1}^{n_i}\sim\mbox{GIG}\left(\alpha^{-1} +2\sum_{j=1}^{n_i} H_0(t_{ij})e^{x_{ij}^\top \beta},\alpha^{-1},\lambda+\sum_{j=1}^{n_i} \delta_{ij}\right). 
\end{eqnarray} 

The $Q$-function is the conditional expected value of the complete log-likelihood function given the observed data at the current estimated parameter values. In other words, we can write 
$$
Q(\theta,\theta^{(r)}) \equiv E_{\theta^{(r)}}\bigg(\ell_c(\theta) \mid \{(t_{ij},\delta_{ij})\}, j=1,\cdots,n_i,\, i=1,\cdots,m \bigg),
$$ 
where $\theta^{(r)}$ denotes the estimated vector of parameters of the EM-algorithm in the $r$-th step. The result depends on the expectations $\omega_i(\theta)\equiv E\left(Z_i | \{(t_{ij},\delta_{ij})\}_{j=1}^{n_i}\right)$ and $\kappa_i(\theta)\equiv E\left(Z_i^{-1} | \{(t_{ij},\delta_{ij})\}_{j=1}^{n_i} \right)$. 
Details about these terms are presented in next proposition.

\begin{proposition} (E-step of the EM algorithm) For $i = 1,\cdots,m$, we have that
	\begin{eqnarray*}
		\omega_i(\theta) = \frac{K_{\lambda + \sum_{j=1}^{n_i} \delta_{ij} +1}\left(\sqrt{\alpha^{-1}[\alpha^{-1} + 2\sum_{j=1}^{n_i} H_0(t_{ij})e^{x_{ij}^\top \beta}]}\right)}{K_{\lambda + \sum_{j=1}^{n_i} \delta_{ij}}\left(\sqrt{\alpha^{-1}[\alpha^{-1} + 2\sum_{j=1}^{n_i} H_0(t_{ij})e^{x_{ij}^\top \beta}]}\right)} \left(1 + 2\alpha\sum_{j=1}^{n_i} H_0(t_{ij})e^{x_{ij}^\top \beta}\right)^{-1/2}
	\end{eqnarray*}
	and
	\begin{eqnarray*}
		\kappa_i(\theta) = \frac{K_{\lambda + \sum_{j=1}^{n_i} \delta_{ij} -1}\left(\sqrt{\alpha^{-1}[\alpha^{-1} + 2\sum_{j=1}^{n_i} H_0(t_{ij})e^{x_{ij}^\top \beta}]}\right)}{K_{\lambda + \sum_{j=1}^{n_i} \delta_{ij}}\left(\sqrt{\alpha^{-1}[\alpha^{-1} + 2\sum_{j=1}^{n_i} H_0(t_{ij})e^{x_{ij}^\top \beta}]}\right)} \left(1+2\alpha\sum_{j=1}^{n_i} H_0(t_{ij})e^{x_{ij}^\top \beta}\right)^{1/2}.
	\end{eqnarray*}
\end{proposition}	
\begin{proof} It follows immediately by using (\ref{moments}) and (\ref{conddist}).
\end{proof}

Using the previous expressions and assuming a piecewise exponential hazard for $h_0(t)$ as stated in (\ref{pwe}), provided a fixed set of cut points, we have fully defined our $Q$-function. The EM-algorithm establishes simpler expressions to be maximized rather than running a direct maximization procedure based on the observed likelihood function in (\ref{loglik_obs}) with a piecewise constant hazard function. In addition, applying the EM approach allows us to do the optimization separately for $(\beta,\eta)$ and $\alpha$. A complete description of the EM-algorithm for the GIG frailty models is given in Algorithm \ref{alg} including initial guesses for initialization of the procedure.

\begin{algorithm}
	\caption{EM-algorithm for the class of PE-GIG frailty models}\label{alg}
	\begin{algorithmic}[1]
		\small
		\STATE Provide initial guesses for the parameters. Denoting $\theta^{(r)}$ as the estimate of the set of parameters at step $r$, $\beta^{(0)}$ will be the coefficients obtained by fitting the proportional hazards model through the \verb|R| package \verb|survival|. Set $\alpha^{(0)} = 1$ and $\eta_k^{(0)} = \sum_{i=1}^{m} \sum_{j=1}^{n_i} \delta_{ij} / \sum_{i=1}^{m} \sum_{j=1}^{n_i} t_{ij}$, if $t_{ij}$ is a failure time for $t^{(l)} \leq t_{ij} < t^{(l+1)}$. The term $t^{(l)}$ represents the $l$-th change point, for $l=1,\cdots,k$. The expression for $\eta_k^{(0)}$ is the maximum likelihood estimator of the rate of an exponential distribution in each given interval using only the failure times within that interval. \newline
		
		\STATE Update the $Q$-function using the expressions found for $\omega_i(\theta)$ and $\kappa_i(\theta)$ with $\theta= \widehat{\theta }^{(r)}$. \newline
		
		\STATE Use the expressions of $h_0(t)$ and $H_0(t)$, based on the piecewise exponential distribution, and numerically obtain the maximum likelihood estimates of $\beta$ and $\eta$ by maximizing $Q_1$. For this task, we chose to work with the \verb|optim| function in \verb|R| under the BFGS maximization routine \citep{fletcher}. The new estimates of the parameters are denoted by ${\widehat{\beta}}^{(r+1)}$ and ${\widehat{\eta}}^{(r+1)}$.  \newline
		
		\STATE Find the maximum likelihood estimate for $\alpha$ by maximizing $Q_2$ numerically. Again, we use \verb|optim| with BFGS. The new estimate is denoted by ${\widehat{\alpha}}^{(r+1)}$. \newline
		
		\STATE Verify a convergence criterion. For example, $||{{\theta}^{(r+1)}-{\theta}^{(r)}||} < \epsilon$ for some pre-established $\epsilon > 0$. If the convergence criterion is satisfied, set ${\theta}^{(r+1)}$ as the final parameter estimates. Otherwise, update ${\theta}^{(r)}$ with ${\theta}^{(r+1)}$ and return to Step 2.
	\end{algorithmic}
\end{algorithm}

In order to obtain the standard errors of the parameter estimates, we choose to work with bootstrap resampling \citep{efron}. The main reason for this strategy is the possibility to avoid cumbersome numerical derivatives that would be necessary to perform the procedure in \cite{klein}, which is based on the observed Fisher Information matrix.


\section{Simulation studies}\label{sec:simulation}

This section shows simulation studies where the main aim it to evaluate the performance, in terms of estimation, of the PE-GIG frailty models under misspecification of the frailty distribution. We study scenarios in which the data are generated with gamma, inverse Gaussian, generalized exponential and log-normal frailties. We fit to each synthetic data set the IG, RIG, HYP and PHYP frailty models with different numbers of change points for the piecewise constant hazard function and run 1000 Monte Carlo replicas in each scenario. We fit the semiparametric versions of the gamma and GE frailty models, described in the literature, to compare with the models proposed in the present paper. 

For the data generated with gamma, generalized exponential and log-normal frailties, we consider the sample sizes of $m = 200$ and $m = 500$ with clusters formed by $n_i = 2$ individuals each, resulting in total sample sizes of $400$ and $1000$. In addition, we explore $\lambda$ misspecification by generating data with inverse-Gaussian frailty and fitting the other special cases. In this context, we consider $m = 20$ and $m = 100$ with $n_i = 10$ for all $i$. This way we maintain the total sample size and are able to test the model under data with larger clusters. In this scenario, we do not report the results of the GE model because the explicit expressions of this model are available for clusters up to size 2.

Given the frailties, the failure times $T_ {ij}^0 $ were generated from a Weibull distribution with $\gamma = 2$ and $\sigma = 0.25$. The censoring times $C_{ij}$ were generated, independently of $T_ {ij}^0 $, following a Weibull distribution with parameters $\gamma = 2$ and $\sigma = 0.05$. The observed data are $T_{ij}=\min\{T_ {ij}^0,C_{ij}\}$ and $\delta_ {ij} = I\{T_{ij}^0 \leq C_{ij} \}$. With this configuration, we have a percentage of approximately 30\% of censored observations. In all scenarios, we consider two covariates that were generated from the Bernoulli$(0.5)$ and Uniform$(-1,1)$ distributions with the true values of their effects being $(\beta_1, \beta_2) = (1.5,-1.0)$. The true value of $\alpha$ was set to be $1$. Fitting the IG, RIG, HYP and PHYP models is done considering $5$ and $10$ change points for the piecewise baseline hazard function. The semiparametric versions of the gamma and GE models evaluated here are based on the Cox partial likelihood function, therefore, they do not require a cut point specification.

The parameter $\alpha$ represents the frailty variance only in the gamma and IG options. Hence, in each case an appropriate transformation of this parameter is calculated so that we obtain the frailty variance. This comparison is done as discussed by \cite{barretosouza}, where it is noted that the model given by $h(t_{ij}|Z_i) = Z_i \ h_0(t_{ij}) \exp\{x_{ij}^\top \beta\}$ is equivalent to $h(t_{ij}|Z_i) = Z_i^{*} \ h_0^{*}(t_{ij}) \exp\{x_{ij}^\top \beta\}$, with $Z_i^* = Z_i/E(Z_i)$ having mean $1$ and $ h_0^{*}(t_{ij}) = h_0(t_{ij})E(Z_i)$. In other words, the comparison of the frailty variance should be done through the transformation $\mbox{Var}(Z_i^*) = \mbox{Var}(Z_i)/E(Z_i)^2$. The proper transformation for each model is done so that they are comparable and is reported in the column named ``Var'' in the forthcoming tables.

In Table \ref{table_GA_200} we present the empirical mean and root mean square error (RMSE) of the estimates of the parameters obtained through 1000 Monte Carlo replicas. Here the data were generated with gamma frailty and sample size $m = 200$, with $n_i = 2$ for all $i$. We fit to this data the gamma, GE and GIG frailty models. 

\begin{table}[]
	\centering
	\caption{Empirical mean and root mean square error (RMSE) of the estimates for $\beta_1$, $\beta_2$ and for the frailty variance of gamma distributed frailty data with sample size $m=200$. Rows represent the fitted model. The true values of the parameters are $\beta_1 = 1.5$, $\beta_2 = -1$ and $\alpha = 1$ (true frailty variance is $1$).} \label{table_GA_200}
	\begin{tabular}{@{}lccccccc@{}}
		\hline
		\multicolumn{2}{c}{\textbf{}}                                    & \multicolumn{2}{c}{\textbf{$\beta_1$}} & \multicolumn{2}{c}{\textbf{$\beta_2$}} & \multicolumn{2}{c}{\textbf{Var}} \\
		\multicolumn{1}{c}{\textbf{Model}}         & \textbf{Cut points} & \textbf{Mean}    & \textbf{RMSE}   & \textbf{Mean}    & \textbf{RMSE}   & \textbf{Mean}   & \textbf{RMSE}  \\ \hline
		Semi. Gamma                       & -                   & 1.288            & 0.212           & $-$0.857           & 0.143           & 0.592           & 0.408          \\
		\\
		Semi. GE                          & -                   & 1.396            & 0.104           & $-$0.927           & 0.073           & 0.767           & 0.233          \\
		\\
		\multirow{2}{*}{PE - IG}   & $k=5$                 & 1.307            & 0.193           & $-$0.862           & 0.138           & 1.541           & 0.541          \\
		& $k=10$                & 1.361            & 0.139           & $-$0.900           & 0.100           & 1.866           & 0.866          \\
		\\
		\multirow{2}{*}{PE - RIG}  & $k=5$                & 1.408            & 0.092           & $-$0.927           & 0.073           & 1.153           & 0.153          \\
		& $k=10$                & 1.482            & 0.018           & $-$0.981           & 0.019           & 1.272           & 0.272          \\
		\\
		\multirow{2}{*}{PE - HYP}  & $k=5$                & 1.365            & 0.135           & $-$0.898           & 0.102           & 1.388           & 0.388          \\
		& $k=10$                & 1.434            & 0.066           & $-$0.946           & 0.054           & 1.608           & 0.608          \\
		\\
		\multirow{2}{*}{PE - PHYP} & $k=5$                 & 1.383            & 0.117           & $-$0.914           & 0.086           & 0.865           & 0.135          \\
		& $k=10$                & 1.427            & 0.073           & $-$0.949           & 0.051           & 0.894           & 0.106          \\ \hline
	\end{tabular}
\end{table}

As can be seen, the correctly specified model clearly underestimates the frailty variance. This is something that was also observed by \cite{barretosouza} in their simulations. In that paper it is pointed that the difficulty in estimating $\alpha$ in the gamma frailty model can be due to the flat shape of its associated $Q$-function. We note here a clear advantage of the PHYP, RIG and GE models in estimating this quantity, as they returned smaller RMSEs than the true model. Although the gamma case is the true model in this scenario, it showed the greatest bias in the estimation of covariate effects. The lowest RMSEs in the estimation of $\beta_1$ and $\beta_2$ are due to GIG frailties. When $k = 10$, the RIG, HYP and PHYP returned excellent average estimation of these coefficients, showing some advantage over the GE model as well.

\begin{table}[]
	\centering
	\caption{Empirical mean and root mean square error (RMSE) of the estimates of $\beta_1$, $\beta_2$ and for the frailty variance of gamma distributed frailty data with sample size $m=500$. Rows represent the fitted model. The real values of the parameters are $\beta_1=1.5$, $\beta_2=-1$ and $\alpha=1$ (true frailty variance is 1).} \label{table_GA_500}
	\begin{tabular}{lccccccc}
		\hline
		\multicolumn{2}{c}{\textbf{}}                                    & \multicolumn{2}{c}{\textbf{$\beta_1$}} & \multicolumn{2}{c}{\textbf{$\beta_2$}} & \multicolumn{2}{c}{\textbf{Var}} \\
		\multicolumn{1}{c}{\textbf{Model}}         & \textbf{Cut points} & \textbf{Mean}    & \textbf{RMSE}   & \textbf{Mean}    & \textbf{RMSE}   & \textbf{Mean}   & \textbf{RMSE}  \\ \hline
		Semi. Gamma                       & -                   & 1.290            & 0.210           & $-0.863$           & 0.137           & 0.609           & 0.391          \\
		\\
		Semi. GE                          & -                   & 1.393            & 0.107           & $-0.929$           & 0.071           & 0.769           & 0.231          \\
		\\
		\multirow{2}{*}{PE - IG}   & $k=5$                & 1.294            & 0.206           & $-0.860$           & 0.140           & 1.480           & 0.480          \\
		& $k=10$                & 1.345            & 0.155           & $-0.897$           & 0.103           & 1.782           & 0.782          \\
		\\
		\multirow{2}{*}{PE - RIG}  & $k=5$                 & 1.395            & 0.105           & $-0.923$           & 0.077           & 1.149           & 0.149          \\
		& $k=10$                & 1.471            & 0.029           & $-0.978$           & 0.022           & 1.275           & 0.275          \\
		\\
		\multirow{2}{*}{PE - HYP}  & $k=5$                 & 1.350            & 0.150           & $-0.894$           & 0.106           & 1.351           & 0.351          \\
		& $k=10$                & 1.418            & 0.082           & $-0.942$           & 0.058           & 1.570           & 0.570          \\
		\\
		\multirow{2}{*}{PE - PHYP} & $k=5$                 & 1.381            & 0.119           & $-0.915$           & 0.085           & 0.882           & 0.118          \\
		& $k=10$                & 1.425            & 0.075           & $-0.951$           & 0.049           & 0.909           & 0.091          \\ \hline
	\end{tabular}
\end{table}

Table \ref{table_GA_500} shows that, in general, increasing the sample size does not lead to a significant decrease of the RMSE in any of the models. Note that the correctly specified frailty model continued to provide the large biases in the estimation of fixed effects and high underestimation of the frailty variance. Regarding the frailty parameter, the model that returned the best results was the PHYP frailty model. Still, the special cases IG, HYP and RIG also displayed excellent estimates of covariate effects, especially when $k = 10$.

In next scenario, data is generated with inverse-Gaussian frailty, one particular case of our proposed class. We have two main goals here: ($i$) evaluate the other particular cases under misspecification of $\lambda $ and ($ii$) explore how the proposed model behaves under larger cluster sizes. For this task, we consider the total sample sizes of 400 and 1000 as before, but $n_i = 10$ for all $i$.

Table \ref{table_IG_200} contains the mean and RMSE of the estimates obtained in 1000 Monte Carlo replicas of this simulation study. Note that, as expected, the true model estimates well the fixed effects and the frailty variance. In this scenario the gamma frailty model displays good estimates of the fixed effects but high RMSE for the frailty variance. In the comparison involving the particular cases of the GIG frailty, one can see that the choice of $\lambda$ does not affect the estimation of the covariate effects, but only the frailty variance. In terms of this quantity, besides the true model, the HYP case is the one presenting the smallest RMSE in the group under comparison. This behavior is also anticipated since the $\lambda$ value corresponding to this special case is zero, being the closest to the true $\lambda$ ($-0.5$; inverse-Gaussian case). We observe that the further the estimate of $\lambda$ is separated from its true value, the greater is the underestimation of the frailty variance. Even so, all GIG frailties estimated this quantity with smaller RMSEs than the gamma model.

\begin{table}[h!]
	\centering
	\caption{Empirical mean and root mean square error (RMSE) of the estimates for $\beta_1$, $\beta_2$ and for the frailty variance of IG distributed frailty data with total sample size equal to 400 ($m=20$ with $n_i = 10 \mbox{ for all } i$). Rows represent the fitted model. The true values of the parameters are $\beta_1 = 1.5$, $\beta_2 = -1$ and $\alpha = 1$ (true frailty variance is $1$).} \label{table_IG_200}
	\begin{tabular}{@{}lccccccc@{}}
		\hline
		\multicolumn{2}{c}{\textbf{}}                                    & 	\multicolumn{2}{c}{\textbf{$\beta_1$}} & \multicolumn{2}{c}{\textbf{$\beta_2$}} & \multicolumn{2}{c}{\textbf{Var}} \\
		\multicolumn{1}{c}{\textbf{Model}}         & \textbf{Cut points} & \textbf{Mean}    & \textbf{RMSE}   & \textbf{Mean}    & \textbf{RMSE}   & \textbf{Mean}   & \textbf{RMSE}  \\ \hline
		Semi. Gamma                       & -                   &     1.471        &    0.029      &  $ -0.980  $    &    0.020      &   0.552      &    0.448      \\
		\\
		\multirow{2}{*}{PE - IG}   & $k=5$                 &  1.450           &   0.050         &   $ -0.962$     & 0.038         &       0.878   &   0.122        \\
		& $k=10$           & 1.487     &     0.014   &  $ -0.992 $      &   0.022      &    0.978       &    0.022     \\
		\\
		\multirow{2}{*}{PE - RIG}  & $k=5$                 &       1.449      &    0.051     &    $ -0.961 $     &    0.039      &      0.657     &      0.343    \\
		& $k=10$                &    1.484       &    0.016        &     $-0.990 $     &    0.010     &     0.702      &   0.304     \\
		\\
		\multirow{2}{*}{PE - HYP}  & $k=5$                 &     1.452        &    0.048    & $ -0.962 $      &   0.038       & 0.755         &    0.245     \\
		& $k=10$                &   1.488        &    0.013      &   $ -0.992 $      &     0.008       &    0.822     &    0.182      \\
		\\
		\multirow{2}{*}{PE - PHYP} & $k=5$                 &       1.444      &   0.056         &  $-0.958$        &  0.042          & 0.573         &   0.427       \\
		& $k=10$           &  1.478    &   0.023   &   $ -0.986$       &   0.015     &     0.603     &       0.406    \\ \hline
	\end{tabular}
\end{table}

Table \ref{table_IG_500} contains the results of a similar simulation study assuming the total sample size being $1000$. We observe a similar behavior in relation to what is shown in Table \ref{table_IG_200}, where all models estimate well the fixed effects and the best estimation of the frailty variance is due to the correctly specified model. Likewise in the gamma scenario, increasing sample size did not produce a significant reduction in RMSE. 

\begin{table}[h!]
	\centering
	\caption{Empirical mean and root mean square error (RMSE) of the estimates of $\beta_1$, $\beta_2$ and of the frailty variance for IG distributed frailty data with total sample size equal to 1000 ($m=100$ with $n_i = 10 \mbox{ for all } i$). Rows represent the fitted model. The real values of the parameters are $\beta_1=1.5$, $\beta_2=-1$ and $\alpha=1$ so real frailty variance is also 1.} \label{table_IG_500}
	\begin{tabular}{@{}lccccccc@{}}
		\hline
		\multicolumn{2}{c}{\textbf{}}                                    & 	\multicolumn{2}{c}{\textbf{$\beta_1$}} & \multicolumn{2}{c}{\textbf{$\beta_2$}} & \multicolumn{2}{c}{\textbf{Var}} \\
		\multicolumn{1}{c}{\textbf{Model}}         & \textbf{Cut points} & \textbf{Mean}    & \textbf{RMSE}   & \textbf{Mean}    & \textbf{RMSE}   & \textbf{Mean}   & \textbf{RMSE}  \\ \hline
		Semi. Gamma                       & -                   &   1.466          &    0.034      &  $ -0.980  $   &  0.020      &   0.557     &   0.443       \\
		\\
		\multirow{2}{*}{PE - IG}   & $k=5$                 &    1.438         &   0.062       &    $ -0.958 $  &   0.042    &    0.867     &    0.133      \\
		& $k=10$           &  1.480    & 0.020       &   $-0.987$       &   0.013     &     0.957       &   0.043      \\
		\\
		\multirow{2}{*}{PE - RIG}  & $k=5$                 &    1.436       &    0.064     &    $-0.956 $    &    0.044    &    0.658       &     0.342    \\
		& $k=10$                &  1.477         &   0.023         &    $-0.985$       &     0.015    &    0.702       &   0.298     \\
		\\
		\multirow{2}{*}{PE - HYP}  & $k=5$                 &  1.439         &   0.061    &   $ -0.958 $   &    0.042    &  0.753        &   0.247     \\
		& $k=10$                &     1.480      &    0.020      &   $-0.987$        &   0.013         &   0.815      &    0.185      \\
		\\
		\multirow{2}{*}{PE - PHYP} & $k=5$                 &    1.432         &   0.068     &    $-0.953 $     &   0.047      &   0.579       &      0.421    \\
		& $k=10$           &   1.471   &   0.029   &    $-0.981$       &    0.019    &    0.608      &    0.392       \\ \hline
	\end{tabular}
\end{table}

Simulation studies were also performed generating data with GE and log-Normal frailty distributions. The results are reported in Appendix B. In the GE scenario, the gamma frailty model showed the highest RMSEs when estimating the two covariate effects. All special cases of the GIG model returned smaller bias in the estimation of $\beta_1$ and $\beta_2$ than those from the correctly specified model when $k = 10$. The best estimate of the frailty variance was due to the GIG special case PHYP. In the log-Normal scenario, we concluded that the gamma and GE frailty models underestimated the frailty variance when the true frailty distribution is log-normal, which was also observed in \cite{barretosouza}. Here the poorest performance is due to the gamma model. Not only it displays the highest bias in the estimation of the frailty variance, but $\beta_1$ and $\beta_2$ also show RMSEs that are considerably higher than those for the other models. All GIG special cases and the GE frailty model yield very good estimates of the covariate effects, but this does not occur for the frailty variance. The model that best estimated this last quantity was the inverse-Gaussian frailty model; it has displayed significantly lower RMSE than all competing cases.

\section{Real data application}\label{sec:analysis}

To illustrate the proposed methodology in this paper, we explore clinical data of children diagnosed with neuroblastoma cancer collected by the Therapeutically Applicable Research to Generate Effective Treatments (TARGET) initiative. Neuroblastoma is a type of cancer that originates in primitive forms of the nerve cells of the sympathetic nervous system. We aim to evaluate the effect of MYCN gene amplification and the presence of hyperdiploid chromosomes in the DNA content of the tumor cell on the survival time of patients. Although the amplification of the MYCN gene is commonly associated with a worse prognosis, the role of hyperdiploid chromosomes in the evolution of tumors still remains an ongoing question in cancer research. The prognosis of the presence of an abnormal number of chromosomes is not well understood as different conclusions have been drawn with respect to different types of cancer. 

For instance, \cite{dastague} associated higher ploidy to a better the prognosis for childhood B-acute lymphoblastic leukemia. In the same vein, \cite{carroll} added that hyperdiploidy seems to favor a beneficial impact in leukemogenesis, possible being a consequence rather than a driver of malignancy. However, \cite{donovan} recognized glioblastoma hyperdiploid tumor cells as a potential contributor to tumor evolution and disease recurrence in adult brain cancer patients. Following this argument, one of our aims in this data application is to reliably assess the effect of hyperdiploid chromosomes on the prognosis of children with neuroblastoma cancer by the usage of  appropriate statistical methodology. 

\subsection{Neuroblastoma data description and preliminary analysis}\label{sec:dataset}

The Therapeutically Applicable Research to Generate Effective Treatments (TARGET) program is intended to study the molecular changes causing childhood cancers. The main goal of this initiative is to collect data allowing researchers to work towards the development of effective and less toxic treatments. The program makes avaliable online data files of multiple cancer types that can be found in the TARGET data matrix section of their website \textcolor{blue}{\url{https://ocg.cancer.gov/programs/target/data-matrix}}; it is necessary to request permission from the National Cancer Institute to use this data, as we did. In this paper, we explore the clinical data of the neuroblastoma cancer containing patient survival data and multiple covariates. As described in their related homepage, the neuroblastoma cancer arises in immature nerve cells of the sympathetic nervous system, primarily affecting infants and children. It accounts for 12\% of childhood cancer mortality, those between 18 months and 5 years of age being the most affected. 

We work with the subset of patients classified as high risk by the Children's Oncology Group (COG), categorization that determines the type of treatment to be received. Of the total of 533 high risk patients with a known vital status, we selected 315 observations that had full information recorded so to test the influence of multiple variables. Among these 315 individuals, 178 of them deceased until the end of the study and the others were right censored. In a preliminary analysis, we assessed the significance of several covariates in the survival time of the patients. Some of them are: The International Neuroblastoma Staging System (INSS stage) that is a clinically and surgically based staging system used to categorize tumor extent; the International Neuroblastoma Pathology Classification Mitosis Karyorrhexis Index Category (MKI); the patient's age at diagnosis; the patient's gender; the International Neuroblastoma Pathology Classification Tumor Cell Differentiation Degree Category, among others. At this stage, we selected covariates using the Cox model and the log-rank test. From this, we kept the variables MYCN gene amplification (categorized and Amplified or Not Amplified) and Ploidy, that is DNA Ploidy Analysis by Flow Cytometry Result Category, described as Diploid (DI=1) or Hyperdiploid (DI>1), with no interaction among them. 

An important and well-known tool to explore time to event data is the Kaplan-Meier estimate of the survival function, introduced in the seminal paper by \cite{km}. In Figure \ref{km_target}, we report this graph for the two covariates selected. We observe that the distance between the curves does not remain proportional over time, indicating that the Cox model is not appropriate for this data application.

\begin{figure}[ht!]
	\centering
	\includegraphics[width=.8\textwidth]{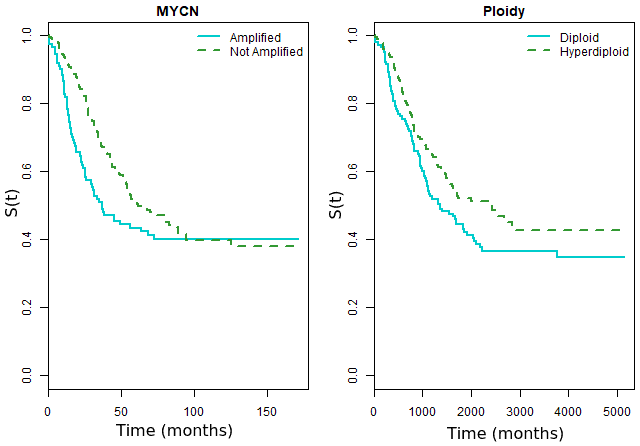} 
	\caption{Descriptive analysis of the TARGET neuroblastoma data set. Kaplan-Meier estimates of the survival function for the variables MYCN and Ploidy.}\label{km_target}
\end{figure}

Our modeling strategy is to apply frailty models where the proportionality of hazards is not suitable. At the same time, such approach allows to account for individual heterogeneity not explained by the covariates. For instance, it is reasonable to expect that biological distinctions (e.g. genetic features) exist among the patients.

In addition to the proposed GIG class, we employ well-established frailty models with implementations available in \texttt{R} program, as would most commonly be done in practice. We seek among the competitive models a robust estimation of the covariate effects and frailty variance. A reliable estimation of this quantities is crucial for our data application. The hazard ratios of MYCN and Ploidy will tell us about the influence of these variables in the patients' survival times, facilitating a prognosis assessment and contributing to the discussion about the role of hyperdiploid chromossomes in cancer cells. In addition, the frailty variance is also of interest as it informs how much heterogeneity to expect among the survival times of high risk neuroblastoma patients.


\subsection{Statistical analysis of neuroblastoma data}

In this section we analyze the TARGET neuroblastoma data using the proposed class of GIG frailties and also well-established frailty models. Our goal is to compare how our models perform in comparison to the frailty models commonly used in practice. Our response variable is the time (in months) from diagnosis to the last follow-up or death of the patient. We work with the subset of patients classified as high risk according to a system developed by the Children's Oncology Group (COG). This data set contains 315 observations, where 178 events are recorded and 137 are censored. We investigate the effects of two covariates on the survival time of the patients. They were selected after the preliminary analysis described in Section \ref{sec:dataset}. The covariates are: ($i$) MYCN Status, it corresponds to MYCN gene amplification status and is categorized as amplified or non-amplified; and ($ii$) Ploidy (DNA Ploidy Analysis by Flow Cytometry Result Category), categorizes the DNA content of the tumor cell as diploid or hyperdiploid.

We compare the estimation of the hazard ratios and frailty variance among the four special cases of the GIG class as well as the gamma, Positive Frailty Variance (PVF) and log-Normal frailties. By using univariate frailty models, we account for non-observed individual risk factors. We focus on semiparametric/quasi-semiparametric implementations of aforementioned models since no information about the parametric form of the baseline risk function is avaliable. Besides, testing the adequacy of parametric forms of this function is not an easy task. We fit the semiparametric implementations of the gamma, PVF and log-normal frailty models available in the \texttt{R} packages \texttt{frailtyEM} \citep{frailtyEM} and \texttt{frailtySurv} \citep{frailtySurv}. 

According to \cite{barretosouza}, the popular Gamma frailty model can present problems in the estimation of the frailty variance depending on the form of the likelihood to be maximized. Trying to circumvent this, two different implementations of this model are investigated here. We fit the versions in \texttt{frailtyEM} and \texttt{frailtySurv} packages. In the first, a non-penalized log-likelihood function is maximized while in the former a pseudo full likelihood approach is employed.

We choose the number of cut points (named $k$) for the GIG class by fitting each special case for a range of $k$ from 3 to 30. We notice that for the IG, HYP, RIG and PHYP frailties, the smallest AIC values were among $k$ from 10 to 15, with small variation in this range. Hence, we choose $k =10$ and present the results with this configuration. 

Table \ref{target_semipar} contains the results of the fitting the aforementioned frailty models to the TARGET neuroblastoma data. In addition to the parameter estimates, we report the exponential of $\beta_{mycn}$ and $\beta_{ploidy}$, i.e. the hazard ratio, which gives us the practical interpretation of the regression coefficients. See \cite{wienke} for a proper interpretation of these hazard ratios for frailty models.

\begin{table}[]
	\caption{
		Estimates of the parameters and standard errors (in parentheses) for the neuroblastoma data under the semiparametric/quasi-semiparametric approach.}\label{target_semipar}
	\begin{tabular}{@{}lllllll@{}}
		\hline
		Model                     & $\beta_{mycn}$ & $\beta_{ploidy}$ & {Var} & $\alpha$  & $e^{\beta_{mycn}}$ & $e^{\beta_{ploidy}}$  \\ \hline
		
		Gamma (\texttt{frailtyEM}) &  1.312 (0.366) &  0.586 (0.341)  &  3.866   &  0.259 (0.100)   & 3.714  &   1.797 \\
		
		Gamma (\texttt{frailtySurv})   &  1.203 (0.907)   &   0.568 (0.224)  &   3.117   &   3.117 (20.996) & 3.329  &  1.765 \\
		
		PVF (\texttt{frailtyEM})   & 1.429 (0.317)    &  0.275 (0.281)  &   2.021   &  0.495 (0.303) &  4.177 & 1.317  \\
		
		LN (\texttt{frailtySurv})   & 0.561 (0.417)    & 0.421 (0.213)   &  40.845   & 1.933 (16.307)  & 1.753  & 1.524 \\
		
		PE-IG ($k=10$)   &   0.463 (0.278)    & 0.426 (0.250)  &    5.939     & 5.939 (0.919) & 1.588 & 1.530  \\
		
		PE-RIG ($k=10$)  & 0.430 (0.272)   & 0.347 (0.241)  & 0.948       & 1.467 (0.093) & 1.537 & 1.415 \\
		
		PE-HYP ($k=10$)  & 0.542 (0.309)   & 0.491 (0.292)  &    3.550    & 7.817 (0.732) & 1.721 & 1.634 \\

		PE-PHYP ($k=10$) & 0.361 (0.249)   & 0.327 (0.212)  &      0.509     & 0.702 (0.029) & 1.435 & 1.388  \\
		\hline
	\end{tabular}
\end{table}

It is immediate that the parameter estimates vary considerably among the fits of the different models. Evidently, there is great variability in the estimation of the frailty variance, but the regression coefficients are also not immune to that. Furthermore, frailty variance related parameter, indicated in the column $\alpha$, displays huge standard errors under the gamma (\texttt{frailtySurv}) and log-normal models. We highlight that in the \texttt{frailtySurv} package standard errors are computed through Bootstrap replication, as are the ones in the GIG class. To provide a fair comparison, the same number of replicas (1000) was used to produce the results in both cases.

Having found a significant variability of results among different frailty models and its implementations, we seek to investigate the robustness of the estimates produced in each case. We propose to examine this through a bootstrap study where, in each replication, we sample 200 out of the 315 observations and fit the models under comparison. After a number of replications, we can assess the variation of the parameter estimates. More specifically, we examine how these vary around the maximum likelihood estimates gathered from the fit to the entire data set. Naturally, the less variability, more consistent the model is and more reliable are the estimates obtained from the complete data. This study is also motivated by the fact that model selection tools are scarce in the frailty models literature. It is specially hard in this case where the competing models employ different forms of the likelihood and distinct forms of the baseline risk functions.

The study is performed using 1000 bootstrap replicas. We report the results in Figures \ref{boot_covs} and \ref{boot_fvar} where each box corresponds to one of the models and the colors to the frailty distribution. Figure \ref{boot_covs} contains the boxplots of the hazard ratio estimates for the covariates Amplified MYCN and Diploid Ploidy. We center those around the values obtained from fit to the complete data set, that we refer as the "true" value.

\begin{figure}[ht!]
	\centering
	\includegraphics[width=1\textwidth]{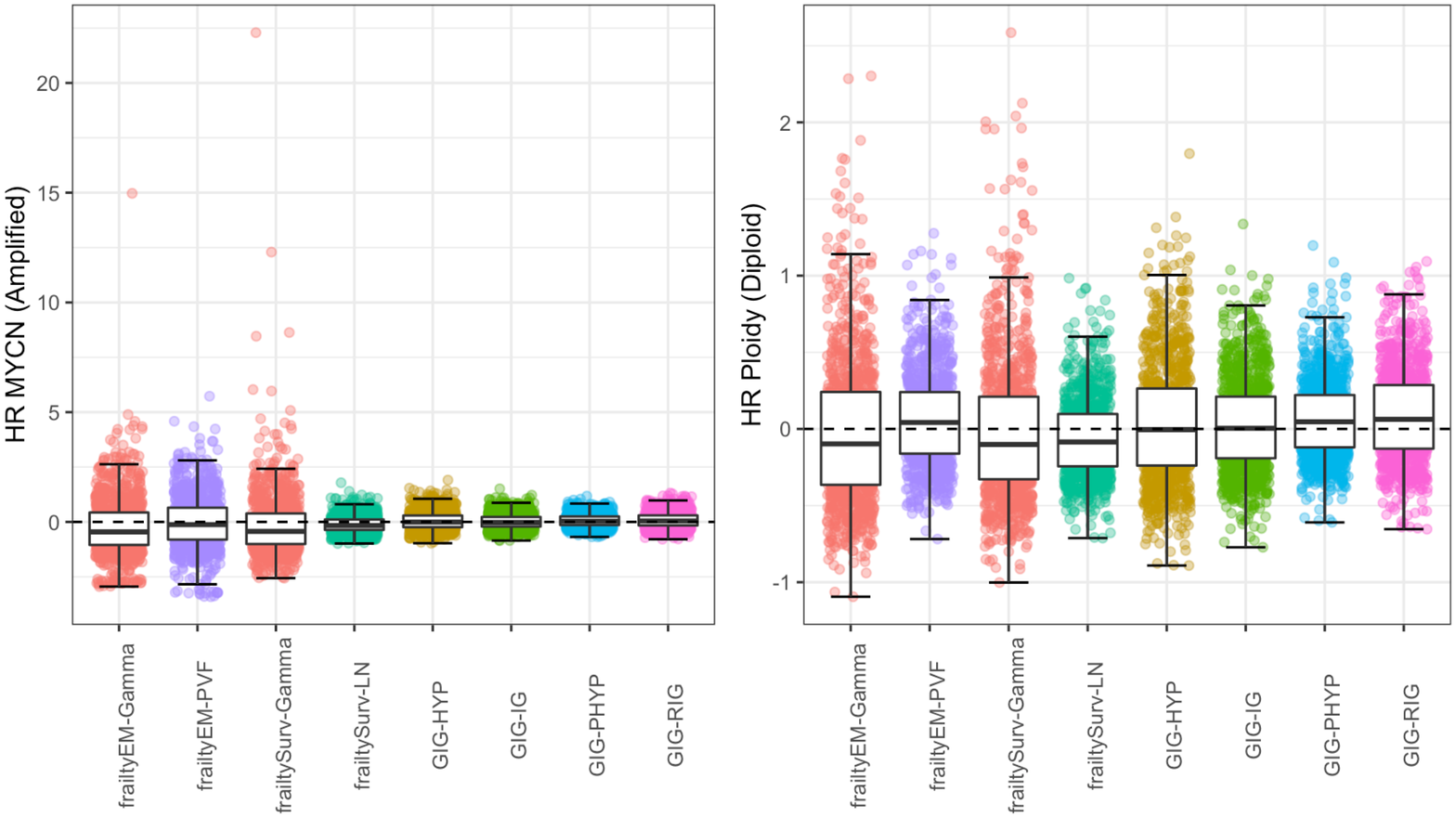} 
	\caption{Results of Bootstrap study for centered hazard ratios to the left. Plot to the right is just a zoom of the left one.}\label{boot_covs}
\end{figure}

In the left-side plot, we can see that there is a greater variability on the estimation of the hazard ratio of MYCN under the fits of the gamma and PVF frailties in comparison to the competing models. Meanwhile, the log-normal and GIG frailties behave well, which is evidenced by their narrower boxes that are concentrated closely to the "true" value. As it concerns Ploidy, the right-hand side plot shows that there is a less striking difference among the models under comparison. Although both gamma fits seem to produce more outliers, in general all the models are reasonably robust when estimating the effect of the Ploidy variable.

\begin{figure}[ht!]
	\centering
	\includegraphics[width=1\textwidth]{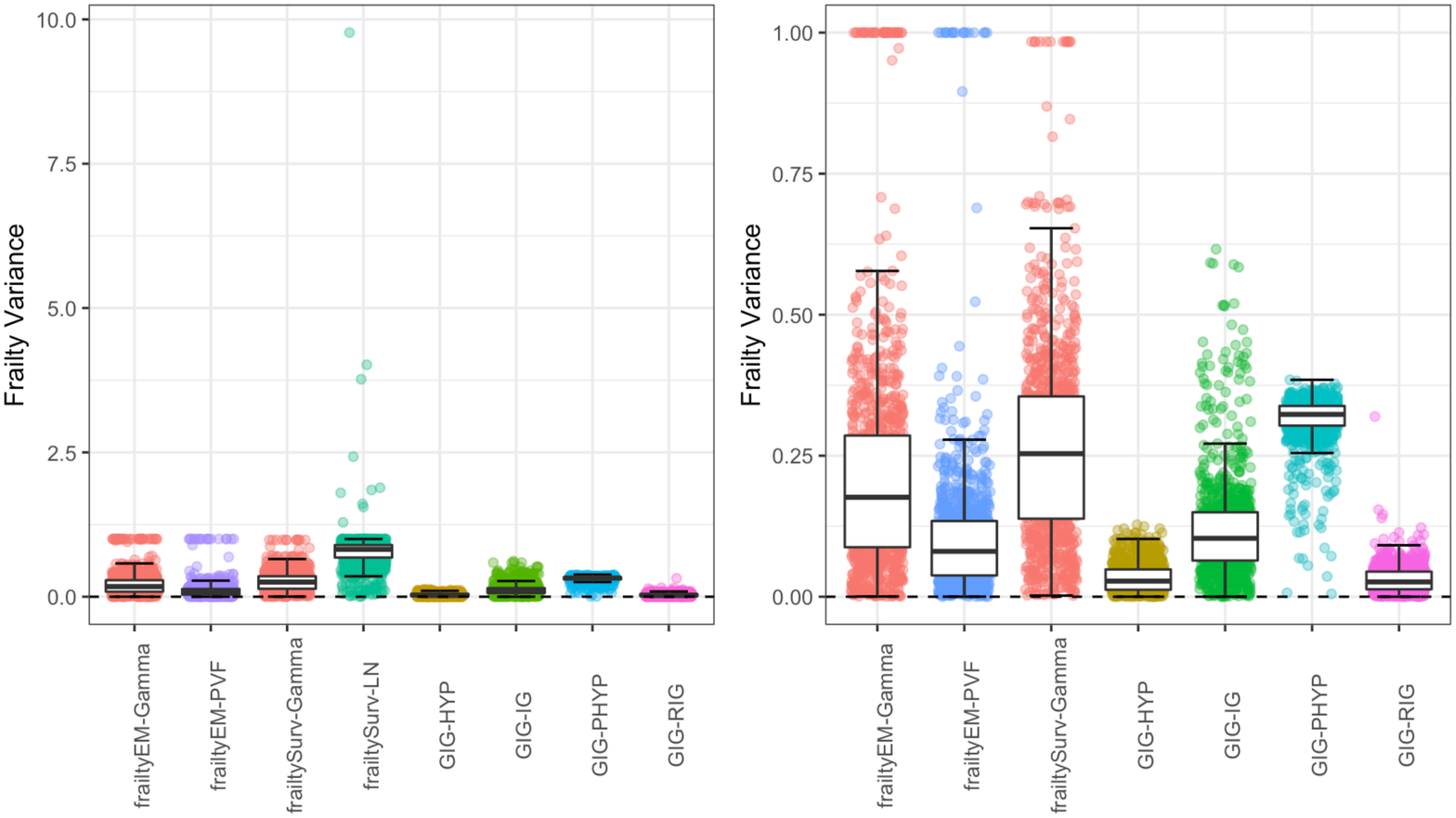} 
	\caption{Results of Bootstrap study for centered and scaled frailty variances for all models to the left. Plot to the right present the same results by excluding the log-normal model.}\label{boot_fvar}
\end{figure}

Figure \ref{boot_fvar} contains the result of this study for the frailty variance. In this case, we center and scale the frailty variance estimates around the "true" values to avoid any interference of the different scales. On the left-hand side, all the models under comparison are included, where the log-normal frailty clearly does not perform like the other ones. It displays high outlier values and is located further from zero then the other models. We filter out the log-normal boxplot from the right-side plot to be able to compare among the other cases. Now, we identify the gamma fits to be less consistent and can points advantages towards the HYP and RIG models whose boxplots are the narrowest and concentrated more closely to zero.

We conclude that the HYP and RIG special cases of the GIG family were the most robust models for estimating the three quantities of interest in the application to the TARGET neuroblastoma data set. They return consistent estimates of both the hazard ratios and frailty variance, something that is not achieved by the competing models. In sum, this study points out that the fits of the HYP or RIG frailty models to the complete data set should be preferred to draw conclusion about the covariate effects and the population heterogeneity in this neuroblastoma data study.

Having shown advantages towards the GIG class in this data application, we propose to make model selection among a broader range of GIG models by considering several values of $\lambda$. Hence, we present the selection of $\lambda$ based on a grid of values. The grid of $\lambda$ is set between $-5$ and $5$ with $0.1$ spacing. We calculate the value of the log-likelihood function for the set of parameter estimates obtained under each different $\lambda$.

\begin{figure}[ht!]
	\centering
	\includegraphics[width=.4\textwidth]{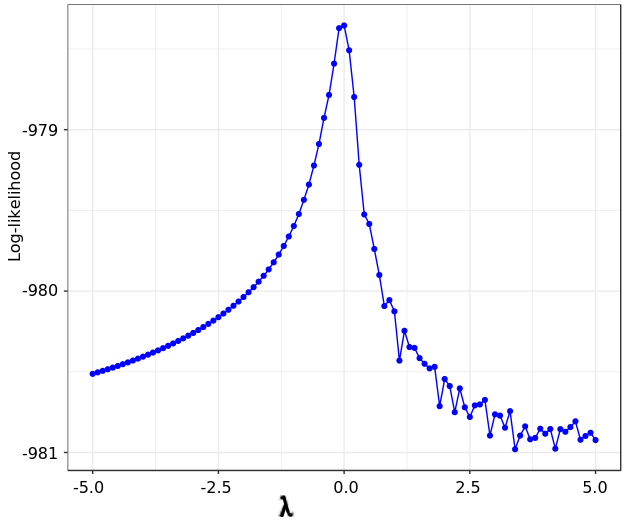} 
	\caption{Profile log-likelihood function of the parameter $\lambda$ for the neuroblastoma data.}\label{lambda_grid}
\end{figure}

The result is reported in Figure \ref{lambda_grid} where it is evident that the value $\lambda = 0$ achieves the highest log-likelihood value. The irregularity of the log-likelihood function on the right is most likely due to the numerical instability of the Bessel function, but this does not compromise our choice for $\lambda=0$ since the function is well behaved around this point. This result goes in accordance with the conclusions drawn from our Boostrap simulation study, where we pointed out the HYP($\lambda =0$) or RIG($\lambda = 0.5$) cases to be preferable.  Furthermore, by combining these results, we can conclude that the most suitable model to the data in terms of model selection is also a robust choice, as desired.

To summarise, our main goal in the application of frailty models to the TARGET neuroblastoma data set was to conclude about the effects of the two genetic covariates and patient heterogeneity. We fitted different frailty models where individual heterogeneity not explained by the covariates was captured by a random effect. A considerable variability in the results when using different frailty distributions was highlighted. We proposed to tackle this by checking the robustness of the estimates through a bootstrap study where we fit each model several times to samples of the original data set. By doing this, it is evidenced that two special cases of the GIG class produce consistent estimates with advantages over the others. Further, we use a profile-likelihood approach to select the most suitable $\lambda$ value. In this analysis, we used an extensive grid of $\lambda$ values to take into consideration a range of GIG models and not only its special cases. The result indicates $\lambda = 0$ to be preferred, which corresponds to the HYP frailty. At the end, we are able to show that the GIG frailty chosen through model selection via profile likelihood is also robust, as ideal.

Hence, we draw our conclusions using parameter estimates of the HYP frailty model. We have that the failure rate of patients with amplified MYCN is $1.721$ times that of patients with the unamplified status, conditional on the same frailty. In addition, the failure rate of patients with diploid chromosomes is $1.634$ times that of patients with hyperdiploid chromosomes, given the frailty value. Further, there is large effect of the patient individual heterogeneity, reflected by the frailty variance of 3.550. The statistical result on the effect of an amplified MYCN gene agrees with what is mentioned in \cite{yoshi}. According to the authors, the MYCN amplification has proven to be an independent prognostic factor for identifying rapid tumor progression and predicting poor prognosis independent of age and clinical stage. As for the presence of hyperdiploid chromosomes in the DNA content of tumor cells, this is a factor in which there is no consensus on its prognostic influence and seems to differ between different types of cancer. Through the analysis of the TARGET neuroblastoma data, we can see that there is evidence that the presence of diploid cells is a protective factor in the lifetime of children with high-risk neuroblastoma cancer, whose biological reasons, evidently, have yet to be investigated.

\section{Discussion}\label{discussion}

In this paper we proposed a flexible class of frailty models based on the generalized inverse-Gaussian (GIG) distributions, which enjoy mathematical tractability like the gamma frailty model. Further, we illustrated through Monte Carlo simulation that the GIG models are robust to the specification of the frailty distribution and by varying $\lambda$ we capture distinct behaviours of the random effect.

Through simulation studies, we were able to evaluate the behavior of the GIG class of frailty models under misspecification. The proposed model proved to be robust as it was able to yield satisfactory estimates of the covariate effects in all scenarios. Under gamma and GE artificial data, there was a particular case of our class that performed better than the correct specified model. It was also possible to obtain a particular case that performed properly when the true frailty was log-normal distributed; this is the scenario in which the competing models returned high biased estimates. In addition, we explored the fit of the GIG special cases under misspecification of $\lambda$. Our findings indicate that, although it affects the estimation of the frailty variance, misspecifying $\lambda$ does not influence largely the estimation of the fixed effects.

We highlight that, as expected, there is no single model that is appropriate in all situations. The frailties obtained by fitting the special cases of the GIG distribution showed different behaviours in the simulation studies. This evidences that they can capture distinct dependence structures providing the desired flexibility. We emphasize the advantage of having a robust and flexible model such as the GIG frailty in hand, because identifying the true frailty distribution is not a simple task in practical problems and using selection and diagnostic methods are still scarce in this field.

A complete statistical analysis was conducted comprising the most commonly used frailty models and the class of GIG frailties with the goal of producing reliable conclusions about the effects of the amplification of the MYCN gene and the presence of hyperdiploid chromosomes on the survival times of children with neuroblastoma cancer. After identifying that the fitted models yielded considerably different results, we compared the robustness of the parameter estimates among them. This was done via a Bootstrap study which showed the HYP and RIG special cases to be the most consistent. A profile likelihood approach was used to select $\lambda$, indicating the choice of the HYP frailty. Having agreement between the results from the robustness study and the profile likelihood approach, we use the fitted HYP frailty model to draw conclusions about the covariate effects and the population heterogeneity. We highlight that this real data application illustrates how our class can be advantageous in practical situations.

Future research includes applying GIG frailty models in other contexts such as current status data. Another point of interest to be attacked is to propose a time-varying GIG frailty model. One possibility of fulfilling this task is to assume a time-varying frailty given by $Z(t)\equiv w(t,Z)$, where $w(t,Z) = Z^{h_0(t)}$, for $t>0$, $Z$ follows a GIG distribution as described in this paper and $h_0(\cdot)$ denotes the baseline hazard function. This structure was introduced by \cite{enketal2014}, with $Z$ following a generalized gamma distribution. As argued by these authors, time-varying frailty models can be applied to study infectious diseases. We hope to explore this point and submit our findings in a future paper.

\section*{Acknowledgments}
We thank the National Cancer Institute (Office of Cancer Genomics) for granting us permission to use the TARGET Neuroblastoma Clinical data for publication.

\newpage

\section*{Appendix}

In this appendix, we provide some additional theoretical and simulated results.

\subsection*{Appendix A: Additional theoretical results}

Consider $\mbox{RFV}(s) = J''(-s/\mu)/J'(-s/\mu)^2$, where $J(s) = \log L(-s)$, $L(\cdot)$ is the Laplace transform and $\mu$ is the expected value of the GIG frailty distribution. The expressions required to calculate the RFV for the GIG frailty model are given by

\begin{eqnarray*} \label{rfv_d1}
	\frac{\partial J(s)}{\partial s} &=& \frac{\partial \log L(-s)}{\partial s} = \frac{\alpha^{-1}}{2(\alpha^{-1}(\alpha^{-1}-2s))^{3/2}K_{\lambda}(\sqrt{\alpha^{-1}(\alpha^{-1}-2s)})} \times \\ && \{  \alpha^{-1}(\alpha^{-1}-2s)K_{\lambda-1}(\sqrt{\alpha^{-1}(\alpha^{-1}-2s)}) \nonumber \\
	&&+2\lambda\sqrt{\alpha^{-1}(\alpha^{-1}-2s)}K_{\lambda}(\sqrt{\alpha^{-1}(\alpha^{-1}-2s)}) \nonumber\\ && +\alpha^{-1}(\alpha^{-1}-2s)K_{\lambda+1}(\sqrt{\alpha^{-1}(\alpha^{-1}-2s)})\}  \nonumber
\end{eqnarray*}
and 
\begin{eqnarray*} \label{rfv_d2}
	\frac{\partial J(s)^2}{\partial s^2}&=& 	\frac{\partial \log L(-s)^2}{\partial s^2}= \frac{1}{4(\alpha^{-1}-2s)^2\sqrt{\alpha^{-1}(\alpha^{-1}-2s)}K_{\lambda}(\sqrt{\alpha^{-1}(\alpha^{-1}-2s)})^2} \times \nonumber \\ && \{-(\alpha^{-1}(\alpha^{-1}-2s))^{3/2}K_{\lambda-1}(\sqrt{\alpha^{-1}(\alpha^{-1}-2s)})^2 +  2\alpha^{-1}(\alpha^{-1}-2s) \times \\ && [K_{\lambda}(\sqrt{\alpha^{-1}(\alpha^{-1}-2s)})-   \sqrt{\alpha^{-1}(\alpha^{-1}-2s)}K_{\lambda+1}(\sqrt{\alpha^{-1}(\alpha^{-1}-2s)})]  \times \nonumber \\ &&K_{\lambda-1}(\sqrt{\alpha^{-1}(\alpha^{-1}-2s)}) - 
	4\alpha^{-1}s\sqrt{\alpha^{-1}(\alpha^{-1}-2s)}K_{\lambda}(\sqrt{\alpha^{-1}(\alpha^{-1}-2s)})^2 +
	\nonumber \\ &&
	8\lambda \sqrt{\alpha^{-1}(\alpha^{-1}-2s)}K_{\lambda}(\sqrt{\alpha^{-1}(\alpha^{-1}-2s)})^2 + \nonumber \\ && 2\alpha^{-2}\sqrt{\alpha^{-1}(\alpha^{-1}-2s)} K_{\lambda}(\sqrt{\alpha^{-1}(\alpha^{-1}-2s)})^2+
	\nonumber \\ &&2\alpha^{-1}s\sqrt{\alpha^{-1}(\alpha^{-1}-2s)}K_{\lambda+1}(\sqrt{\alpha^{-1}(\alpha^{-1}-2s)})^2 - \nonumber\\ &&\alpha^{-2}\sqrt{\alpha^{-1}(\alpha^{-1}-2s)}K_{\lambda+1}(\sqrt{\alpha^{-1}(\alpha^{-1}-2s)})^2 + \nonumber\\ &&
	(\alpha^{-1}(\alpha^{-1}-2s))^{3/2}K_{\lambda-2}(\sqrt{\alpha^{-1}(\alpha^{-1}-2s)})K_{\lambda}(\sqrt{\alpha^{-1}(\alpha^{-1}-2s)})+\nonumber\\
	&&2\alpha^{-2}K_{\lambda}(\sqrt{\alpha^{-1}(\alpha^{-1}-2s)})K_{\lambda+1}(\sqrt{\alpha^{-1}(\alpha^{-1}-2s)})-
	\nonumber\\ &&4\alpha^{-1}sK_{\lambda}(\sqrt{\alpha^{-1}(\alpha^{-1}-2s)})K_{\lambda+1}(\sqrt{\alpha^{-1}(\alpha^{-1}-2s)})-  \nonumber \\
	&&2\alpha^{-1}s\sqrt{\alpha^{-1}(\alpha^{-1}-2s)}K_{\lambda}(\sqrt{\alpha^{-1}(\alpha^{-1}-2s)})K_{\lambda+2}(\sqrt{\alpha^{-1}(\alpha^{-1}-2s)})+\nonumber \\
	&&\alpha^{-2}\sqrt{\alpha^{-1}(\alpha^{-1}-2s)}K_{\lambda}(\sqrt{\alpha^{-1}(\alpha^{-1}-2s)})K_{\lambda+2}(\sqrt{\alpha^{-1}(\alpha^{-1}-2s)})\}.
\end{eqnarray*}

\begin{lemma} \label{lemma1}
	Let $\zeta(x) = K_\phi(\sqrt{x}) / x^{\phi/2}$, we have that $\dfrac{\partial^k}{\partial x^k} \zeta(x) =  \left(-\dfrac{1}{2} \right)^{k} \dfrac{K_{\phi+k} (\sqrt{x})}{x^{(\phi+k)/2}}$, for $k\in\mathbb N$.
\end{lemma}

\begin{proof}
	Case $k=1$.
	It can be shown that the derivative of the Bessel function with respect to its argument is
	\begin{eqnarray*}
		\frac{\partial}{\partial x}K_\phi(\sqrt{x}) = -\frac{1}{4\sqrt{x}}\left(K_{\phi+1}  (\sqrt{x})+ K_{\phi-1}  (\sqrt{x})\right).
	\end{eqnarray*}
	Using this expression, we have that
	\begin{eqnarray*}
		\frac{\partial}{\partial x}\zeta(x)= -\frac{1}{2} \left( \phi\frac{K_\phi(\sqrt{x})}{x^{(\phi/2 + 1)}} + \frac{K_{\phi+1}  (\sqrt{x})+ K_{\phi-1}  (\sqrt{x})  }{2x^{(\phi+1)/2}}    \right).     
	\end{eqnarray*}
	At this point, we apply the recurrence identity on the Bessel function previously mentioned,
	\begin{eqnarray} \label{bessel_id}
	K_v(z) = \frac{z}{2v}\left(K_{v+1}(z) - K_{v-1}(z)\right),
	\end{eqnarray}
	and the former derivative simplifies to
	\begin{eqnarray*}
		\frac{\partial}{\partial x}\zeta(x)=  -\frac{1}{2} \dfrac{K_{\phi+1}  (\sqrt{x})}{x^{(\phi+1)/2}}.
	\end{eqnarray*}
	Using this, we continue the demonstration by finding the second derivative of $\zeta(x)$ with respect to $x$. The result is as follows
	\begin{eqnarray*}
		\frac{\partial^2}{\partial x^2}\zeta(x)= -\frac{1}{2} \left( -\frac{K_{\phi+1}  (\sqrt{x})(\phi+1)}{2x^{(\phi+3)/2}} - \frac{K_{\phi+2}  (\sqrt{x})+ K_{\phi}  (\sqrt{x})     }{4x^{(\phi/2 +1)}}          \right).
	\end{eqnarray*}
	Using (\ref{bessel_id}), the second derivative simplifies to 
	\begin{eqnarray*}
		\frac{\partial^2}{\partial x^2}\zeta(x)=  \frac{1}{4} \frac{K_{\phi+2}  (\sqrt{x})}{x^{(\phi+2)/2}}.
	\end{eqnarray*}
	To finish the proof by induction, we now assume as true the case $k-1$ and use it to prove the $k$-th order expression. If this is satisfied, then the result is true for all $k$. If the $\{k-1\}$-th derivative of $\zeta(x)$ with respect to $x$ is given by
	\begin{eqnarray*}
		\frac{\partial^{k-1}}{\partial x^{k-1}}\zeta(x)=  \left( {-\dfrac{1}{2}}\right)^{k-1}  \frac{K_{\phi+k-1}  (\sqrt{x})}{x^{(\phi+k-1)/2}} ,
	\end{eqnarray*}
	then the $k$-th order derivative is
	\begin{eqnarray*}
		\frac{\partial^{k-1}}{\partial x^{k-1}}\zeta(x)=  \left( {-\dfrac{1}{2}}\right)^{k-1} \left\{  \frac{-K_{\phi+k-1}  (\sqrt{x})(\phi +k-1)}{2x^{(\phi+k+1)/2}} -\frac{K_{\phi+k}  (\sqrt{x})+ K_{\phi+k-2}  (\sqrt{x})  }{4x^{(\phi+k)/2}}           \right\}.
	\end{eqnarray*}
	Here we use again (\ref{bessel_id}), which provides $K_{\phi+k-1}(\sqrt{x}) = \frac{\sqrt{x}}{2(\phi+k-1)}[K_{\phi+k}(\sqrt{x})-K_{\phi+k-2}(\sqrt{x})]$
	and the result
	\begin{eqnarray*}
		\frac{\partial^k}{\partial x^k}\zeta(x)=  \left(-\dfrac{1}{2}\right)^{k} \frac{K_{\phi+k}  (\sqrt{x})}{x^{(\phi+k)/2}}.
	\end{eqnarray*}
	This completes the proof of Lemma 1.
\end{proof}

\subsection*{Appendix B: Additional simulated results}

Here, we present Monte Carlo simulations for the GIG frailty model under the quasi-semiparametric framework based on the piecewise exponential (PE) baseline hazard function to study model misspecification. The true frailty distribution is the generalized exponential (GE) and log-normal. In both scenarios we explore the sample sizes of $m = 200$ and $m = 500$ with clusters formed by $n_i = 2$ individuals each. The number of Monte Carlo replications is 1000. The simulation results for GE data are presented in Tables \ref{table_GE_200} and \ref{table_GE_500}, where we provide the empirical mean and the root of the mean squared errors (RMSE) of the estimates for the covariate effects and for the frailty variance. Results related to the log-normal scenario are exhibited in Tables \ref{table_LN_200} and \ref{table_LN_500}. The data generation procedure adopted here is described in Section 6 of the main paper.

\begin{table}[]
	\centering
	\caption{Empirical mean and root mean square error (RMSE) for $\beta_1$, $\beta_2$ and for the frailty variance of GE distributed frailty data with sample size $m=200$. Rows represent the fitted model. The real values of the parameters are $\beta_1=1.5$, $\beta_2=-1$ and $\alpha=1$, so real frailty variance is also 1.} \label{table_GE_200}
	\begin{tabular}{@{}lccccccc@{}}
		\toprule
		\multicolumn{2}{c}{\textbf{}}                                    & 	\multicolumn{2}{c}{\textbf{$\beta_1$}} & \multicolumn{2}{c}{\textbf{$\beta_2$}} & \multicolumn{2}{c}{\textbf{Var}} \\
		\multicolumn{1}{c}{\textbf{Model}}         & \textbf{Cut points} & \textbf{Mean}    & \textbf{RMSE}   & \textbf{Mean}    & \textbf{RMSE}   & \textbf{Mean}   & \textbf{RMSE}  \\ \midrule
		Semi. Gamma                       & -                   & 1.304            & 0.196           & $-0.869$           & 0.131           & 0.645           & 0.355          \\
		\\
		Semi. GE                          & -                   & 1.391            & 0.109           & $-0.924$           & 0.076           & 0.770           & 0.230          \\
		\\
		\multirow{2}{*}{PE - IG}   & $k=5$                 & 1.335            & 0.166           & $-0.885$           & 0.115           & 1.602           & 0.602          \\
		& $k=10$                & 1.407            & 0.093           & $-0.938$           & 0.062           & 2.064           & 1.064          \\
		\\
		\multirow{2}{*}{PE - RIG}  & $k=5$                 & 1.423            & 0.077           & $-0.938$           & 0.062           & 1.163           & 0.163          \\
		& $k=10$                & 1.483            & 0.017           & $-0.982$           & 0.018           & 1.257           & 0.257          \\
		\\
		\multirow{2}{*}{PE - HYP}  & $k=5$                 & 1.388            & 0.112           & $-0.917$           & 0.083           & 1.407           & 0.407          \\
		& $k=10$                & 1.448            & 0.052           & $-0.960$           & 0.040           & 1.595           & 0.595          \\
		\\
		\multirow{2}{*}{PE - PHYP} & $k=5$                 & 1.393            & 0.107           & $-0.921$           & 0.079           & 0.876           & 0.124          \\
		& $k=10$               & 1.426            & 0.074           & $-0.949 $          & 0.051           & 0.895           & 0.105          \\ \midrule
	\end{tabular}
\end{table}

When data are generated with GE frailty, we can see in Table \ref{table_GE_200} that the gamma model displays the higher RMSEs in the estimation of $\beta_1$ and $\beta_2$. All the GIG frailties show lower RMSEs than the gamma and GE cases. In addition, the correctly specified model and the gamma model underestimate the frailty variance. The PHYP and RIG frailties perform well in this scenario, estimating with low bias the frailty variance and the covariate effects. 

\begin{table}[]
	\centering
	\caption{Empirical mean and root mean square error (RMSE) for $\beta_1$, $\beta_2$ and for the frailty variance of GE distributed frailty data with sample size $m=500$. Rows represent the fitted model. The real values of the parameters are $\beta_1=1.5$, $\beta_2=-1$ and $\alpha=1$, so real frailty variance is also 1.}  \label{table_GE_500}
	\begin{tabular}{@{}lccccccc@{}}
		\toprule
		\multicolumn{2}{c}{\textbf{}}                                    & \multicolumn{2}{c}{\textbf{$\beta_1$}} & \multicolumn{2}{c}{\textbf{$\beta_2$}} & \multicolumn{2}{c}{\textbf{Var}} \\
		\multicolumn{1}{c}{\textbf{Model}}         & \textbf{Cut points} & \textbf{Mean}    & \textbf{RMSE}   & \textbf{Mean}    & \textbf{RMSE}   & \textbf{Mean}   & \textbf{RMSE}  \\ \midrule
		Semi. Gamma                       & -                   & 1.311            & 0.189           & $-0.876$           & 0.124           & 0.655           & 0.345          \\
		\\
		Semi. GE                          & -                   & 1.394            & 0.106           & $-0.928$           & 0.072           & 0.771           & 0.229          \\
		\\
		\multirow{2}{*}{PE - IG}   & $k=5$                & 1.397            & 0.103           & $-0.925$           & 0.075           & 3.683           & 2.683          \\
		& $k=10$                & 1.382            & 0.118           & $-0.921$           & 0.079           & 1.847           & 0.847          \\
		\\
		\multirow{2}{*}{PE - RIG}  & $k=5$                 & 1.417            & 0.083           & $-0.936  $         & 0.064           & 1.159           & 0.159          \\
		& $k=10$                & 1.479            & 0.021           & $-0.982$           & 0.018           & 1.259           & 0.259          \\
		\\
		\multirow{2}{*}{PE - HYP}  & $k=5$                 & 1.382            & 0.118           & $-0.914 $          & 0.086           & 1.376           & 0.376          \\
		& $k=10$                & 1.443            & 0.057           & $-0.957 $          & 0.043           & 1.565           & 0.565          \\
		\\
		\multirow{2}{*}{PE - PHYP} & $k=5$                 & 1.397            & 0.103           & $-0.925  $         & 0.075           & 0.892           & 0.108          \\
		& $k=10$                & 1.431            & 0.069           & $-0.954 $          & 0.046           & 0.910           & 0.090          \\ \midrule
	\end{tabular}
\end{table}

In Table \ref{table_GE_500} we can see that the same behaviour is observed and that increasing the sample size did not reduce significantly the RMSE in any of the model considered in this simulation study. We can conclude that in this scenario the most competitive model is the PHYP frailty, that had the smallest RMSE in the estimation of the frailty variance among all competing models. Nevertheless, all GIG frailties estimated the fixed effects with low bias.


In Table \ref{table_LN_200}, we present the simulation study for the configuration with log-normal frailty distribution. The total sample size is 400. Unlike the gamma and GE models, where $\alpha = 1$ implies in a frailty variance of 1, in the log-normal case the variance is approximately 1.718 when $\alpha = 1$. We refer again to \cite{barretosouza}, that noted in their work that the gamma and GE models considerably underestimate the frailty variance when the frailty distribution is log-normal. We observed in this scenario that the poorest performance is related to the gamma model. Note that this case not only displays the highest bias in the estimation of the frailty variance, but $\beta_1$ and $\beta_2$ also show RMSEs that are considerably larger than the competing frameworks.

\begin{table}[]
	\centering
	\caption{Empirical mean and root mean square error (RMSE) for $\beta_1$, $\beta_2$ and for the frailty variance of LN distributed frailty data with sample size $m=200$. Rows represent the fitted model. The real values of the parameters are $\beta_1=1.5$, $\beta_2=-1$ and $\alpha=1$, so real frailty variance is approximately 1.718.}  \label{table_LN_200}
	\begin{tabular}{@{}lccccccc@{}}
		\toprule
		\multicolumn{2}{c}{\textbf{}}                                    & \multicolumn{2}{c}{\textbf{$\beta_1$}} & \multicolumn{2}{c}{\textbf{$\beta_2$}} & \multicolumn{2}{c}{\textbf{Var}} \\
		\multicolumn{1}{c}{\textbf{Model}}         & \textbf{Cut points} & \textbf{Mean}    & \textbf{RMSE}   & \textbf{Mean}    & \textbf{RMSE}   & \textbf{Mean}   & \textbf{RMSE}  \\ \midrule
		Semi. Gamma                       & -                   & 1.335            & 0.165           & $-0.887 $          & 0.113           & 0.442           & 1.276          \\
		\\
		Semi. GE                          & -                   & 1.493            & 0.007           & $-0.988 $          & 0.012           & 0.739           & 0.979          \\
		\\
		\multirow{2}{*}{PE - IG}   & $k=5$                 & 1.400            & 0.100           & $-0.926 $          & 0.074           & 1.124           & 0.594          \\
		& $k=10$                & 1.464            & 0.036           & $-0.971$           & 0.029           & 1.421           & 0.297          \\
		\\
		\multirow{2}{*}{PE - RIG}  & $k=5$                 & 1.407            & 0.093           & $-0.928 $          & 0.072           & 0.782           & 0.936          \\
		& $k=10$                & 1.466            & 0.034           &$ -0.972$           & 0.028           & 0.884           & 0.834          \\
		\\
		\multirow{2}{*}{PE - HYP}  & $k=5$                 & 1.413            & 0.087           & $-0.933 $          & 0.067           & 0.947           & 0.771          \\
		& $k=10$                & 1.478            & 0.022           & $-0.979$           & 0.021           & 1.121           & 0.597          \\
		\\
		\multirow{2}{*}{PE - PHYP} & $k=5$               & 1.390            & 0.110           & $-0.918 $          & 0.082           & 0.648           & 1.070          \\
		& $k=10$                & 1.439            & 0.061           & $-0.955$           & 0.045           & 0.704           & 1.014          \\  \midrule
	\end{tabular}
\end{table}

In the log-normal frailty context, the model that best estimated the frailty variance was the inverse Gaussian. When $k = 10$, this model presents a significantly lower RMSE than all remaining contenders. The gamma case, which is the most commonly used in practice, estimated very poorly the true log-normal frailty distribution and also showed the highest RMSEs for $\beta_1$ and $\beta_2$.

Table \ref{table_LN_500} contains the results for an increased total sample size of $1000$. We can see that the gamma model remains not competitive in this scenario, displaying the highest errors in the estimation of all quantities under comparison. Furthermore, the IG frailty continued to show great advantage in the estimation of the frailty variance compared to the other models when adopting 10 change points.

\begin{table}[]
	\centering
	\caption{Empirical mean and root mean square error (RMSE) for $\beta_1$, $\beta_2$ and for the frailty variance of LN distributed frailty data with sample size $m=500$. Rows represent the fitted model. The real values of the parameters are $\beta_1=1.5$, $\beta_2=-1$ and $\alpha=1$, so real frailty variance is approximately 1.718.} \label{table_LN_500}
	\begin{tabular}{@{}lccccccc@{}}
		\toprule
		\multicolumn{2}{c}{\textbf{}}                                    & \multicolumn{2}{c}{\textbf{$\beta_1$}} & \multicolumn{2}{c}{\textbf{$\beta_2$}} & \multicolumn{2}{c}{\textbf{Var}} \\
		\multicolumn{1}{c}{\textbf{Model}}         & \textbf{Cut points} & \textbf{Mean}    & \textbf{RMSE}   & \textbf{Mean}    & \textbf{RMSE}   & \textbf{Mean}   & \textbf{RMSE}  \\ \midrule
		Semi. Gamma                       & -                   & 1.332            & 0.168           & $-0.891$           & 0.109           & 0.454           & 1.264          \\
		\\
		Semi. GE                          & -                   & 1.485            & 0.015           & $-0.988$           & 0.012           & 0.741           & 0.977          \\
		\\
		\multirow{2}{*}{PE - IG}   & $k=5$                 & 1.385            & 0.115           & $-0.921$           & 0.079           & 1.080           & 0.638          \\
		& $k=10$                & 1.451            & 0.049           & $-0.967$           & 0.033           & 1.366           & 0.352          \\
		\\
		\multirow{2}{*}{PE - RIG}  & $k=5$                 & 1.390            & 0.110           & $-0.922  $         & 0.078           & 0.774           & 0.944          \\
		& $k=10$                & 1.451            & 0.049           & $-0.965$           & 0.035           & 0.880           & 0.837          \\
		\\
		\multirow{2}{*}{PE - HYP}  & $k=5$                & 1.397            & 0.103           &$ -0.928 $          & 0.072           & 0.926           & 0.792          \\
		& $k=10$                & 1.463            & 0.037           & $-0.974$           & 0.026           & 1.100           & 0.618          \\
		\\
		\multirow{2}{*}{PE - PHYP} & $k=5$                 & 1.377            & 0.123           & $-0.914 $          & 0.086           & 0.652           & 1.066          \\
		& $k=10$                & 1.428            & 0.072           & $-0.952$           & 0.048           & 0.713           & 1.005          \\ \midrule
	\end{tabular}
\end{table}


\end{document}